\documentclass[11pt,letterpaper]{article}
\usepackage{times}
\usepackage{amsmath,amsfonts,amssymb,amsthm}
\usepackage[colorlinks=true,linkcolor=red,citecolor=blue,urlcolor=red]{hyperref}
\usepackage{paralist}
\usepackage{framed}
\usepackage{verbatim}

\usepackage{mathptmx}
\usepackage[text={6.45in,8.95in}]{geometry}

\newenvironment{reminder}[1]{\smallskip
\noindent {\bf Reminder of #1  }\em}{}

\newtheorem{theorem}{Theorem}[section]
\newtheorem{hypothesis}{Hypothesis}

\newtheorem{proposition}{Proposition}

\newtheorem{lemma}{Lemma}[section]

\newenvironment{proofof}[1]{\smallskip
\noindent {\bf Proof of #1.  }}{\hfill$\Box$
\smallskip}

\def \sgn {\textrm{sign}}

\def \MODp {{\sf MODp}}
\def \MODm {{\sf MODm}}

\def \AC {{\sf AC}}
\def \THR {{\sf THR}}
\def \ETHR {{\sf ETHR}}
\def \NP {{\sf NP}}

\def \TIME {{\sf TIME}}

\def \NTIME {{\sf NTIME}}

\def \MOD {{\sf MOD}}
\def \AND {{\sf AND}}

\def \XOR {\MOD2}

\def\eps{\varepsilon}

\def\poly{\text{poly}}

\def \E {\text{\sf E}}
\def \P {\text{\sf P}}

\def \ACC {{\sf ACC}}
\def \Z {{\mathbb Z}}
\def \R {{\mathbb R}}
\def \F {{\mathbb F}}

\def \NC {{\sf NC}}

\def \Q {{\mathbb Q}}
\newcommand{\ip}[2]{\ensuremath{\left\langle #1,#2\right\rangle}}
\def \LIN {{\sf SUM}}

\title{Limits on representing Boolean functions\\ by linear combinations of simple functions:\\ thresholds, ReLUs, and low-degree polynomials} 

\author{R. Ryan Williams\footnote{EECS and CSAIL, MIT. Supported by NSF CCF-1553288. Any opinions, findings and conclusions or recommendations expressed in this material are those of the authors and do not necessarily reflect the views of the National Science Foundation.} 
}

\begin{document}
\date{}
\maketitle

\begin{abstract} We consider the problem of representing Boolean functions exactly by ``sparse'' linear combinations (over $\R$) of functions from some ``simple'' class ${\cal C}$. In particular, given ${\cal C}$ we are interested in finding low-complexity functions lacking sparse representations. When ${\cal C}$ is the set of PARITY functions or the set of conjunctions, this sort of problem has a well-understood answer; the problem becomes interesting when ${\cal C}$ is ``overcomplete'' and the set of functions is not linearly independent. We focus on the cases where ${\cal C}$ is the set of linear threshold functions, the set of rectified linear units (ReLUs), and the set of low-degree polynomials over a finite field, all of which are well-studied in different contexts. 

Building on the new easy witness lemma of Cody Murray and the author, we provide generic tools for proving lower bounds on representations of this kind. Applying these, we give several new lower bounds for ``semi-explicit'' Boolean functions. Let $\alpha(n)$ be an unbounded function such that $n^{\alpha(n)}$ is time constructible (e.g. $\alpha(n) = \log^{\star}(n)$). We show:
\begin{itemize}
\item Functions in $\NTIME[n^{\alpha(n)}]$ that require super-polynomially many linear threshold functions to represent (depth-two neural networks with sign activation function, a special case of depth-two threshold circuit lower bounds). 
\item Functions in $\NTIME[n^{\alpha(n)}]$ that require super-polynomially many ReLU gates to represent (depth-two neural networks with ReLU activation function).
\item Functions in $\NTIME[n^{\alpha(n)}]$ that require super-polynomially many $O(1)$-degree $\F_p$-polynomials to represent exactly, for every prime $p$ (related to problems regarding Higher-Order Uncertainty Principles). We also obtain a function in $\E^{\NP}$ requiring $2^{\Omega(n)}$-size linear combinations.
\item Functions in $\NTIME[n^{\poly(\log n)}]$ that require super-polynomially many $\ACC \circ \THR$ circuits of polynomial size to represent exactly (further generalizing the recent lower bounds of Murray and the author).
\end{itemize}
We also obtain ``fixed-polynomial'' lower bounds for functions in $\NP$, for the first three representation classes.
\end{abstract}

\thispagestyle{empty}
\newpage
\setcounter{page}{1}

\section{Introduction}
Given $f : \{0,1\}^n\rightarrow \{0,1\}$ and a class ${\cal C}$ of ``simple'' functions, when can $f$ be represented exactly as a \emph{short} $\R$-linear combination of functions from ${\cal C}$? When ${\cal C}$ forms a basis for $B_n$ (the set of all Boolean functions on $n$ inputs) the question has a unique answer that is generally easy to obtain, by analyzing the appropriate linear system (the cases where ${\cal C}$ is the set of all parity functions or the set of all conjunctions are canonical examples). For $|{\cal C}| \gg 2^n$, the situation becomes much more interesting, as there can be many possible representations. The general problem of understanding which functions do and do not have sparse representations for simple ${\cal C}$ arises in many different mathematical topics. Three relevant to TCS are depth-two threshold circuits, depth-two neural networks with various activation functions, and higher-order Fourier analysis. We use the notation \[\LIN \circ {\cal C}\] to denote the class of $\R$-linear combinations of ${\cal C}$-functions; for example, $\LIN \circ \XOR$ denotes $\R$-linear combinations of PARITY functions. The relevant complexity measure for a ``circuit'' in $\LIN \circ {\cal C}$ is the fan-in of the $\LIN$ gate, which we call the \emph{sparsity} of the circuit. 

\paragraph{Sums of Threshold Circuits.} Let $\LIN \circ \THR$ be linear combinations of linear threshold functions (LTFs).\footnote{From here on, ``linear combination'' means ``$\R$-linear combination'', unless otherwise specified.} As there are $2^{\Theta(n^2)}$ $n$-variate threshold functions~\cite{Winder62}, a function $f: \{0,1\}^n \rightarrow \{0,1\}$ has \emph{many} possible representations as a $\LIN \circ \THR$. Such circuits are also known in the machine learning literature as \emph{depth-two neural networks with sign activation functions}. 

In 1994, Roychowdhury, Orlitsky, and Siu~\cite{RoychowdhuryOS94} noted that no interesting size lower bounds were known for computing Boolean functions with $\LIN \circ \THR$ circuits (beyond the few that are/were known for $\THR \circ \THR$~\cite{Hajnal93,RoychowdhuryOS94,KaneW16,ChenSS15,Tamaki16,AlmanCW16}). The problem was raised again more recently in CCC'10 by Hansen and Podolskii~\cite{HP10}. In particular, the following remains largely unanswered:

{\narrower

{\bf Problem:} {\em Find an explicit $f : \{0,1\}^{\star} \rightarrow \{0,1\}$ without polynomially-sparse $\LIN \circ \THR$, i.e., every linear combination of LTFs computing $f$ on $n$-bit inputs needs $n^{\omega(1)}$ LTFs, for infinitely many $n$.}

}

Because of prior lower bounds in weaker settings (such as majority-of-majority~\cite{Hajnal93} and majority-of-thresholds~\cite{Nisan94}), it is natural to think that correlation bounds against linear threshold functions should help.\footnote{That is, one wants to show that a function cannot be $(1/2+\eps(n))$-approximated by a linear threshold function, for the tiniest $\eps(n) > 0$ possible.} Correlation bounds do imply lower bounds for $\LIN \circ \THR$, but only when the weights in the linear combination are not too large (i.e., the weights must be in $[-2^{\delta n},2^{\delta n}]$ for small $\delta < 1$). However, if arbitrary weights are allowed, interesting lower bounds on $\LIN \circ \THR$ (beyond $\Omega(n^{2.5}$ wires~\cite{KaneW16}) were open, to the best of our knowledge. In Section~\ref{section-SUM-THR}, we prove arbitrary polynomial lower bounds for $\NP$  functions: 

\begin{theorem} \label{NP-LIN-THR} For all $k$, there is an $f_k \in \NP$ without $\LIN \circ \THR$ circuits of $n^k$ sparsity. Furthermore, for every unbounded $\alpha(n)$ such that $n^{\alpha(n)}$ is time constructible, there is a function in $\NTIME[n^{\alpha(n)}]$ that does not have $\LIN \circ \THR$ circuits of polynomial sparsity.
\end{theorem}

Note that for arbitrary circuits (even for $\THR \circ \THR$ circuits) the best known complexity for such functions without $n^k$-size circuits (for fixed $k$) is $\sf MA/1$ (\cite{Santhanam09}) and $S_2^p$.

\paragraph{Sums of ReLU Gates.} A ReLU (rectified linear unit) gate is a function $f : \{0,1\}^t \rightarrow \R^+$ such that there is a vector $w \in \R^t$ and scalar $a \in \R$ such that for all $x$, \[f(x) = \max\{0,\langle x,w\rangle + a\}.\] It is important to note that ReLU gates might not be Boolean-valued, but they must output non-negative numbers on all Boolean inputs. Linear combinations of ReLU gates are also known as \emph{depth-two neural networks with ReLU activation functions}, and they are intensely studied in machine learning. Several lower bounds for Sums-of-ReLU functions (which for consistency we call $\LIN \circ {\sf ReLU}$) have recently been shown for functions with \emph{real-valued} inputs and outputs (examples include~\cite{EldanS16,Telgarsky16,arora2016understanding,Daniely17,safran2017depth}) but none of the methods extend to Boolean functions, to the best of our knowledge. Recently, Mukherjee and Basu~\cite{Basu-Mukherjee17} have proved $\Omega(n^{1-\delta})$-gate lower bounds for $\THR\circ {\sf ReLU}$ circuits computing the Andreev function, extending ideas in~\cite{KaneW16}.

Observing that for $|\langle x,w\rangle| \geq 1$ we have \[\max\{0,\langle x,w\rangle + 1\}-\max\{0,\langle x,w\rangle\} = \sgn(\langle x,w\rangle),\] it follows that every $\LIN\circ\THR$ circuit can be simulated by a $\LIN\circ {\sf ReLU}$ circuit with only a doubling of the sparsity. In Section~\ref{section-ReLU} we extend our lower bounds to Sums-of-ReLU circuits:

\begin{theorem}\label{NP-ReLU}
For all $k$, there is an $f_k \in \NP$ without $\LIN \circ {\sf ReLU}$ circuits of $n^k$ sparsity. Furthermore, for every unbounded $\alpha(n)$ such that $n^{\alpha(n)}$ is time constructible, there is a function in $\NTIME[n^{\alpha(n)}]$ that does not have $\LIN \circ {\sf ReLU}$ circuits of polynomial sparsity. 
\end{theorem}

\paragraph{Representing Boolean Functions With Higher-Order Polynomials.} Higher-order Fourier analysis of Boolean functions deals with representing Boolean functions by $\R$-linear combinations of $\F_2$-polynomials of degree higher than one (see~\cite{HigherOrderSurvey16} for a survey of some applications in CS theory). The question of which (if any) explicit functions lack \emph{sparse} representations, even for degree-two polynomials, has been wide open. Letting $\MOD2$ be the class of parity functions, this question asks to find lower bounds for $\LIN \circ \MOD2 \circ \AND_2$ circuits (in our notation, $\AND_k$ denotes ANDs of fan-in at most $k$). Such lower bound problems appear much more difficult than the degree-one case of $\LIN \circ \MOD2$. Even understanding the sparsity of the $\AND$ function in the quadratic (and in general, degree-$O(1)$) setting is a prominent open problem:

\begin{hypothesis}[Quadratic Uncertainty Principle~\cite{SimonsOpenProblems14}]\label{deg-2-uncertainty} There is an $\eps > 0$ such that the $\AND$ function on $n$ variables does not have $\LIN \circ \MOD2 \circ \AND_2$ circuits of $2^{\eps n}$ sparsity.
\end{hypothesis}

Although it is believed that $\AND$ needs exponential sparsity, to our knowledge the \emph{only} lower bound known for an explicit function in $\LIN \circ \MOD2 \circ \AND_2$ was $\Omega(n)$-sparsity. For completeness we include a proof provided to us by Lovett~\cite{Lovett-personal17}) in Appendix~\ref{AND-linear-LB}. Again, when the weights in the linear combination are required to be small (magnitudes are $2^{\eps n}$ for small $\eps > 0$), correlation bounds yield some results: one example (among many) is the work of Green~\cite{Green04} showing that a majority vote of quadratic $\F_3$-polynomials needs $2^{\Omega(n)}$ polynomials to compute PARITY. (Other works in this vein include~\cite{HastadG91,CaiGT96,Bourgain05,GalT12}; see Viola~\cite{Viola09} for a survey.) However, for arbitrary weights, no non-trivial lower bounds have been reported (to our knowledge).

In Section~\ref{section-low-degree}, we prove polynomial sparsity lower bounds for Boolean functions in $\NP$ and $2^{\Omega(n)}$-size lower bounds for $\E^{\NP}$, against linear combinations of polynomials over any prime field with any constant degree: 

\begin{theorem} \label{NP-polys} For every integer $k,d \geq 1$ and prime $p$, there is an $f_k \in \NP$ without $\LIN \circ \MODp \circ \AND_d$ circuits of $n^k$ sparsity. Furthermore, for every unbounded $\alpha(n)$ such that $n^{\alpha(n)}$ is time constructible, there is a function in $\NTIME[n^{\alpha(n)}]$ that does not have $\LIN \circ \MODp \circ \AND_d$ circuits of polynomial sparsity.
\end{theorem}

\begin{theorem} \label{ENP-polys} For every $d \geq 1$ and prime $p$, there is an $\alpha > 0$ and an $f \in \E^{\NP}$ without $\LIN \circ \MODp \circ \AND_d$ circuits of $2^{\alpha n}$ sparsity.  
\end{theorem}

Note the ``smallest'' known complexity class for a function lacking $2^{\Omega(n)}$-size circuits is $\E^{\Sigma_2 \P}$~\cite{MiltersenVW99}, and it is a longstanding open problem to reduce the complexity class for such a function, even against depth-3 $\AC^0$ circuits.

\subsection{Intuition}

Here we give an overview of some of the ideas used to prove the lower bounds in this work. The lower bounds of this paper follow the high-level strategy of proving circuit lower bounds by designing circuit-analysis (satisfiability) algorithms~\cite{Williams10,WilliamsJACM14,WilliamsTHR14}. However, in this work we must execute this strategy differently. All previous lower bounds proved in this framework utilize the ``polynomial method'' from circuit complexity in various ways (representing a circuit by a low-degree polynomial of some kind), combined with fast matrix multiplication and/or fast polynomial evaluation. These approaches do not seem to work for solving SAT on linear combinations of thresholds, low-degree polynomials, or ReLU gates. For example, we do not know how to get a sparse (probabilistic or approximate) polynomial (over any field) for computing an OR of many $\LIN \circ \THR$s, and it is likely that any reasonable approach via polynomials would fail to yield non-trivial results. However, we are able to adapt some bits of the polynomial method to the setting of low-degree polynomials (see Section~\ref{section-low-degree}). 

Another complication is that, in the prior lower bound arguments, a nondeterministic procedure \emph{guesses} a small circuit $C$ of the kind one wishes to prove a lower bound against, and composes $C$ with other Boolean circuitry to form a SAT instance. In our case, if we guess some  arbitrary $\LIN \circ {\cal C}$ circuit, we first need to know if this circuit is actually computing a Boolean function; if not, then the satisfiability question itself is not well-defined, and it will not be possible to meaningfully compose such a circuit with other Boolean circuits. Thus we need a way to efficiently check whether a linear combination is Boolean-valued.

We give a generic way to ``lift'' non-trivial algorithms for counting SAT assignments to short products of ${\cal C}$ circuits to non-trivial algorithms for detecting if a given $\LIN \circ {\cal C}$ circuit is Boolean-valued and for counting SAT assignments. More precisely, we show that in order to prove lower bounds for linear combinations of ${\cal C}$-functions, it suffices to solve a certain sum-product task faster than exhaustive search:

{\narrower

{\bf Sum-Product over ${\cal C}$:} Given $k$ functions $f_1,\ldots,f_k$ from ${\cal C}$, each on Boolean variables $x_1,\ldots,x_n$, compute \[\sum_{x \in \{0,1\}^n} \prod_{i=1}^k f_i(x).\]

}

Note the Sum-Product is computed over $\R$, and the task makes sense even if the functions $f_1,\ldots,f_k$ output \emph{non-Boolean} values. Further note that if the functions $f_1,\ldots,f_k$ are Boolean-valued, then the product of $k$ of them is simply the $\AND$ of $k$ of them. In general, the Sum-Product problem will be $\NP$-hard for most interesting representation classes: for example, it is already equivalent to Subset Sum when ${\cal C}$ is the set of \emph{exact} threshold functions (see Section~\ref{prelims} for a definition). Our meta-theorem states that mild improvements over exhaustive search for Sum-Product over ${\cal C}$ imply strong lower bounds for $\LIN \circ {\cal C}$:

\begin{theorem}\label{generic-LBs}
Suppose every $C \in {\cal C}$ has a $\poly(n)$-bit representation, where each $C$ can be evaluated on a given input in $\poly(n)$ time. Assume there is an $\eps > 0$ and for $k=1,\ldots,4$ there is an $n^{O(1)}\cdot 2^{n-\eps n}$-time algorithm for computing the Sum-Product of $k$ functions $f_1(x_1,\ldots,x_n),\ldots,f_k(x_1,\ldots,x_n)$ from ${\cal C}$. Then:
\begin{compactenum}
\item For every $k$, there is a function in $\NP$ that does not have $\LIN \circ {\cal C}$ circuits of sparsity $n^k$. 
\item For every unbounded $\alpha(n)$ such that $n^{\alpha(n)}$ is time constructible, there is a function in $\NTIME[n^{\alpha(n)}]$ that does not have $\LIN \circ {\cal C}$ circuits of polynomial sparsity. 
\end{compactenum}
\end{theorem}

Theorem~\ref{generic-LBs} is used to prove lower bounds against $\LIN \circ \THR$, $\LIN \circ {\sf ReLU}$, and $\LIN \circ \MOD2 \circ \AND_{O(1)}$. For the $\E^{\NP}$ lower bounds, we use a closure property of $\LIN \circ \MOD2 \circ \AND_{O(1)}$ combined with standard ideas from this line of work (see Theorem~\ref{generic-LBs2}). 

Theorem~\ref{generic-LBs} (and its components) can also be used to easily ``lift'' existing circuit lower bounds to \emph{linear combinations} of those circuits:

\begin{theorem} \label{lin-acc-thr} For every $d,m \geq 1$, there is a $b \geq 1$ and an $f \in {\sf NTIME}[n^{\log^b n}]$ that does not have $\LIN \circ \AC^0_d[m] \circ \THR$ circuits of $n^a$ size, for every $a$.
\end{theorem}

That is, we obtain super-polynomial sparsity lower bounds on representing nondeterministic quasi-polynomial-time functions with $\R$-linear combinations of $\ACC \circ \THR$ circuits (each of polynomial size). This applies the fact that we can solve the Sum-Product problem on $\ACC \circ \THR$ circuits (because we can count SAT assignments to them), with an analogous running time as the best SAT algorithm. More details on Theorem~\ref{lin-acc-thr} can be found in Section~\ref{section-meta-theorem}.

\paragraph{Outline.} The next section is the Preliminaries, which gives background knowledge. Section~\ref{section-meta-theorem} proves Theorem~\ref{generic-LBs}. In Sections~\ref{section-SUM-THR}, \ref{section-ReLU}, and \ref{section-low-degree}, Sum-Product algorithms for $\THR$, ${\sf ReLU}$, and $\MODp \circ \AND_d$ (degree-$d$ $\F_p$-polynomials) are provided which beat exhaustive search. The algorithms for $\THR$ and ${\sf ReLU}$ (Theorems~\ref{sum-prod-THR} and \ref{sum-prod-ReLU}) build upon and extend old Subset-Sum algorithms (Theorem~\ref{subset-sum}). The algorithm for $\MODp \circ \AND_d$ (Theorem~\ref{sum-prod-polys}) uses tools from the polynomial method in a new way. Applying Theorem~\ref{generic-LBs} to each of these algorithms, we obtain strong lower bounds for $\LIN \circ {\cal C}$ for all three classes ${\cal C}$. 

\section{Preliminaries}\label{prelims}

Let ${\cal C}$ be a class of functions of the form $f : \{0,1\}^n \rightarrow \R$. Each member $C \in {\cal C}$ has a number of inputs $n$ and a size, which is the length of the representation of $C$ in bits. 
For the classes $\THR$, $\MOD2 \circ \AND_{O(1)}$, and ${\sf ReLU}$, the size $|C|$ of a representation is $\poly(n)$ bits, without loss of generality; see Proposition~\ref{prop-weights}. (For classes such as $\MOD2 \circ \AND_{\log_2(n)}$, a member of the class takes $\Omega(n^{\log n})$ bits to represent, in the worst case.) We assume that for all $n$, our class ${\cal C}$ contains the projection functions $f_i(x_1,\ldots,x_n) = x_i$ for all $i=1,\ldots,n$. We also assume that ${\cal C}$ is \emph{evaluatable}, meaning that there is a universal $k \geq 1$ such that every $C \in {\cal C}$ can be evaluated on a given input in $O(|C|^k)$ time. All classes we consider have this property.

As is standard, we let ${\sf ANY}_c$ denote the class of Boolean functions with $c$ inputs (the class contains ``any'' such function).

An arbitrary $\LIN \circ {\cal C}$ circuit $C$ over $n$ variables represents some function $f : \{0,1\}^n \rightarrow \R$. We say that $C$ is \emph{Boolean-valued} if for all $x \in \{0,1\}^n$, the output of $C$ on $x$ is in $\{0,1\}$. The following proposition is useful to keep in mind, as it shows that every sparse linear combination of Boolean functions implementing another Boolean function has an equivalent linear combination with ``reasonable'' coefficients.

\begin{proposition} \label{prop-weights} Let ${\cal C}$ be a class of functions with co-domain $\{0,1\}$, and let $C$ be an $\LIN \circ {\cal C}$ circuit of sparsity $s$ that is Boolean-valued. There is an equivalent $\LIN \circ {\cal C}$ circuit $C'$ such that every weight in the linear combination of $C'$ has the form $j/k$, where both $j$ and $k$ are integers in $[-s^{s/2},s^{s/2}]$. 
\end{proposition}

\begin{proof} (See also \cite{Muroga-Toda-Takasu61,Babai-weights10}.) Let $C$ be a linear combination of $s$ functions from ${\cal C}$. WLOG, the set of $s$ Boolean functions from ${\cal C}$ is a linearly independent set (otherwise, we could obtain a smaller linear combination representing the same function). The problem of finding coefficients for the Boolean-valued $C$ is equivalent to solving a certain linear system $Ax=b$ in $s$ unknowns over the rationals, where $b \in \{0,1\}^{2^n}$ and $A \in \{0,1\}^{s \times 2^n}$. Take a linearly independent subsystem of $s$ of these $2^n$ equations. Since the determinant of any $s \times s$ Boolean matrix is in $[-s^{s/2},s^{s/2}]$~\cite{Hadamard1893}, the result follows from Cramer's rule.
\end{proof}

The relevant theorem for sums of ReLU gates is more involved, but Maass~\cite{Maass97} shows how the weights for a circuit of size $s$ need only $\poly(s,n)$ bits of precision. Such ``analog-to-digital'' results are crucial for our work, as in our lower bound proofs we will need a discrete nondeterministic algorithm to \emph{guess} a $\LIN \circ {\cal C}$ circuit and check various properties of it.

\paragraph{Useful Results For Thresholds.} We draw from several algorithms and representation theorems from past work. For $\LIN \circ \THR$, we eventually appeal to a classic result from exact algorithms:

\begin{theorem}[Horowitz and Sahni~\cite{Horowitz-Sahni74}] \label{subset-sum} The number of Subset Sum solutions to any arbitrary instance of $n$ items with integer weights of magnitude $[-2^W,2^W]$ can be computed in $2^{n/2}\cdot \poly(W)$ time.
\end{theorem}

Theorem~\ref{subset-sum} is usually stated in terms of \emph{finding} a subset sum solution, but the algorithm can be easily adapted to count solutions as well.

A Boolean function $f$ is called an \emph{exact threshold function} if there are real-valued $\alpha_1,\ldots,\alpha_n$ and $t$ such that for all $x$, \[f(x) = 1 ~\iff~ \sum_i \alpha_i x_i = t.\] Let $\ETHR$ be the class of exact threshold functions. For our $\LIN \circ \THR$ circuit results, the following transformation is extremely useful: 

\begin{theorem}[Hansen and Podolskii~\cite{HP10}] \label{THR2ETHR} Every linear threshold function in $n$ variables can be represented as an linear combination of $\poly(n)$ exact threshold functions, each with coefficient $1$.
\end{theorem}

It follows that every $\LIN \circ \THR$ of sparsity $s$ has an equivalent $\LIN \circ \ETHR$ of sparsity $\poly(s)$. The idea is that a $\THR$ function defines a set of points in the Boolean hypercube lying on one side of a given hyperplane; we can ``cover'' all the points lying on one side by a \emph{disjoint} sum of $\poly(n)$ ``parallel'' hyperplanes, which function as $\ETHR$ gates. Thus each coefficient in the linear combination is simply $1$.

Another useful property of $\ETHR$ gates is that they are closed under AND:

\begin{theorem}[Hansen and Podolskii~\cite{HP10}] \label{ANDETHR2ETHR} Every conjunction of $t$ exact threshold functions in $n$ variables with integer weights in $[-W,W]$ can be converted in $\poly(t,n)$ time to an equivalent \emph{single} exact threshold gate, with weights in $[-(nW)^{\Theta(t)},(nW)^{\Theta(t)}]$.
\end{theorem}

The idea is simple: if we multiply the $i$th exact threshold gate's linear form by the factor $(nW)^i$, no linear form will ``interfere'' with the other sums, and we can determine if all of them are satisfied simultaneously with one exact threshold. 

\paragraph{Useful Results for Finite Field Polynomials.} Two tools from the literature will be helpful for our results on linear combinations of polynomials. The first is \emph{modulus-amplifying polynomials}, which have been used in Toda's Theorem~\cite{Toda91}, representations of $\ACC$ and $\ACC$-SAT algorithms~\cite{Beigel-Tarui,WilliamsJACM14}, algorithms for All-Pairs Shortest Paths~\cite{ChanW16}, and algorithms for solving polynomial systems~\cite{LokshtanovPTWY17}:

\begin{lemma}[Beigel and Tarui~\cite{Beigel-Tarui}]\label{mod-amplifying}
For all $\ell \in \Z^+$, the degree-$(2\ell-1)$ polynomial (over $\Z$)
\[P_{\ell}(y)=1-(1-y)^{\ell}\sum_{j=0}^{\ell-1}\binom{\ell+j-1}{j}y^j
\] has the property for all integers $m \geq 2$, 
\begin{compactitem}
\item if $y=0 \bmod m$ then $P_{\ell}(y)=0 \bmod m^{\ell}$, 
\item if $y=1 \bmod m$ then $P_{\ell}(y)=1 \bmod m^{\ell}$. 
\end{compactitem}
Furthermore, each coefficient in $F_{\ell}$ has magnitude at most $2^{O(\ell)}$.
\end{lemma}

Recall that a multivariate polynomial is \emph{multilinear} if it contains no powers larger than one. The second tool is a classic result on rapidly evaluating a multilinear polynomial on all points in the Boolean hypercube. 

\begin{theorem}[cf.~\cite{BjorklundHK09}, Section 2.2] \label{poly-eval} Given the $2^n$-coefficient vector of a multilinear polynomial $p \in \Z[x_1,\ldots,x_n]$ where each coefficient is in $[-W,W]$, the value of $p$ on all points in $\{0,1\}^n$ can be computed in $2^n \cdot \poly(n,\log W)$ time.
\end{theorem}

The algorithm of Theorem~\ref{poly-eval} can be obtained by divide-and-conquer (as described in~\cite{WilliamsSIGACT11}) or by dynamic programming (as in~\cite{BjorklundHK09}, Section 2.2).

\paragraph{Connections Between Nondeterministic Circuit UNSAT Algorithms and Circuit Lower Bounds.} We also appeal to several known connections between circuit UNSAT algorithms that beat exhaustive search and circuit lower bounds against nondeterministic time classes, which build on prior work~\cite{Williams10,JMV13,SanthanamWilliams13,Ben-Sasson-Viola14}. 

\begin{theorem}[\cite{Murray-Williams17}]\label{NP-from-UNSAT} If there is an $\eps > 0$ such that Circuit Unsatisfiability for (fan-in 2) circuits with $n$ inputs and $2^{\eps n}$ size is solvable in $O(2^{n-\eps n})$ nondeterministic time, then for every $k$ there is a function in $\NP$ that does not have $n^k$-size (fan-in 2) circuits. 
\end{theorem}

\begin{theorem}[Corollary 12 in Tell~\cite{Tell18}, following \cite{Murray-Williams17}]\label{NP-from-UNSAT2} If there is a $\delta > 0$ and $c \geq 1$ such that Circuit Unsatisfiability for (fan-in 2) circuits with $n$ variables and $m$ gates is solvable in $O(2^{n(1-\delta)}\cdot m^c)$ nondeterministic time, then for every unbounded $\alpha(n)$ such that $n^{\alpha(n)}$ is time-constructible, there is a function in $\NTIME[n^{\alpha(n)}]$ that is not in $\P/\poly$.
\end{theorem}

\begin{theorem}[\cite{Murray-Williams17}]\label{NQP-from-UNSAT} If there is an $\eps > 0$ such that Circuit Unsatisfiability for (fan-in 2) circuits with $n$ inputs and $2^{n^{\eps}}$ size is solvable in $O(2^{n-n^{\eps}})$ nondeterministic time, then for every $k$ there is a function in $\NTIME[n^{\poly(\log n)}]$ that does not have $n^{\log^k n}$-size (fan-in 2) circuits. 
\end{theorem}

In fact, all of these algorithms-to-lower-bounds connections still hold when we replace Circuit Unsatisfiability with the promise problem of distinguishing unsatisfiable circuits from circuits with $2^{n-1}$ satisfying assignments. 

\paragraph{The Power of Linear Combinations of Low-Degree Polynomials.} We note that classical work suggests that $\R$-linear combinations of higher-degree $\F_2$-polynomials can be quite powerful. For example, applying Valiant's depth reduction~\cite{Valiant77} and using the representation of the AND function in the Fourier basis, it is easy to show that every $O(n)$-size $O(\log n)$-depth circuit can be represented by a linear combination of $2^{O(n/\log \log n)}$ $\F_2$-polynomials of degree $O(n^{\eps})$, for any desired $\eps > 0$. Moreover, one can represent any $O(n)$-size ``Valiant series-parallel'' circuit (see \cite{Calabro08}) by a linear combination of $2^{\eps n}$ $\F_2$-polynomials of degree $2^{2^{O(1/\eps)}}$. Hence there is a natural barrier to proving exponential-sparsity lower bounds for linear combinations of ``somewhat-low'' degree polynomials.

\section{Meta-Theorem for Lower Bounds on Linear Combinations of Simple Functions}\label{section-meta-theorem}

In this section, we prove our generic theorem which is applied in subsequent sections to prove lower bounds against linear combinations of threshold functions, ReLU gates, and constant-degree polynomials. Recall (from the Introduction) the Sum-Product problem:

{\narrower

{\bf Sum-Product over ${\cal C}$:} Given $k$ functions $f_1,\ldots,f_k$ from ${\cal C}$, each on Boolean variables $x_1,\ldots,x_n$, compute \[\sum_{x \in \{0,1\}^n} \prod_{i=1}^k f_i(x).\]

}

\begin{reminder}{Theorem~\ref{generic-LBs}}
Suppose every $C \in {\cal C}$ has a $\poly(n)$-bit representation, where each $C$ can be evaluated on a given input in $\poly(n)$ time. Assume there is an $\eps > 0$ and for $k=1,\ldots,4$ there is an $n^{O(1)}\cdot 2^{n-\eps n}$-time algorithm for computing the Sum-Product of $k$ functions $f_1(x_1,\ldots,x_n),\ldots,f_k(x_1,\ldots,x_n)$ from ${\cal C}$. Then:
\begin{compactenum}
\item For every $k$, there is a function in $\NP$ that does not have $\LIN \circ {\cal C}$ circuits of sparsity $n^k$. 
\item For every unbounded $\alpha(n)$ such that $n^{\alpha(n)}$ is time constructible, there is a function in $\NTIME[n^{\alpha(n)}]$ that does not have $\LIN \circ {\cal C}$ circuits of polynomial sparsity. 
\end{compactenum}
\end{reminder}

The remainder of this section is devoted to proving Theorem~\ref{generic-LBs}, and an extension to $\E^{\NP}$ in some cases. We are able to use much of the earlier arguments~\cite{Williams10,WilliamsJACM14,Murray-Williams17} as black boxes. However we need several modifications. 

The first new component needed is a method for checking that a given linear combination of ${\cal C}$ circuits actually encodes a Boolean function (i.e. is Boolean-valued on all Boolean inputs). This is provided by the following theorem:

\begin{theorem} \label{check-boolean} Assume there is an $\eps > 0$ and for $k=1,\ldots,4$ there is an $n^{O(1)}\cdot 2^{n-\eps n}$-time algorithm for computing the Sum-Product of $k$ functions $f_1(x_1,\ldots,x_n),\ldots,f_k(x_1,\ldots,x_n)$ from ${\cal C}$.\\ Then there is an $2^{n-\eps n} \cdot \poly(n,s)$-time algorithm that, given $f(x_1,\ldots,x_n)$ which is an arbitrary linear combination of $s$ functions from ${\cal C}$, determines whether or not $f(a) \in \{0,1\}$ for all $a \in \{0,1\}^n$. 
\end{theorem}

\begin{proof} Suppose we are given $f = \sum_{i=1}^s \alpha_i c_i$, where $\alpha_i \in \R$ and $c_i \in {\cal C}$ each have $n$ inputs. Consider the polynomial \[h(x) := f(x)^2\cdot (1-f(x))^2 = f(x)^2 - 2f(x)^3 + f(x)^4.\] Observe that:
\begin{compactitem}
\item If $f(a) \in \{0,1\}$ for all $a \in \{0,1\}^n$, then $h(a) = 0$ for all $a$.
\item $f(b) \notin \{0,1\}$ implies $h(b) > 0$.  
\item For all $a \in \{0,1\}^n$, $h(a) \geq 0$.
\end{compactitem}
Therefore $\sum_{a \in \{0,1\}^n} h(a) = 0$ if and only if $f(a) \in \{0,1\}$ for all $a \in \{0,1\}^n$.
By applying the distributive law to each of $f(x)^2$, $f(x)^3$, $f(x)^4$, and exchanging the order of summation, we have
\begin{align*}
\sum_{a \in \{0,1\}^n} h(a) &= \sum_{i_1,i_2} \beta_{i_1,i_2} \left(\sum_{a \in \{0,1\}^n} f_{i_1}(x)\cdot f_{i_2}(x)\right)\\
& ~~ + \sum_{i_1,i_2,i_3} \gamma_{i_1,i_2,i_3} \left(\sum_{a \in \{0,1\}^n} f_{i_1}(x)\cdot f_{i_2}(x)\cdot f_{i_3}(x)\right)\\
& ~~ + \sum_{i_1,i_2,i_3,i_4} \delta_{i_1,i_2,i_3,i_4} \left(\sum_{a \in \{0,1\}^n} f_{i_1}(x)\cdot f_{i_2}(x)\cdot f_{i_3}(x)\cdot f_{i_4}(x)\right)
\end{align*}
for $\beta_{i_1,i_2} = \alpha_{i_1}\cdot \alpha_{i_2}$, $\gamma_{i_1,i_2,i_3} = -2\alpha_{i_1}\cdot \alpha_{i_2}\cdot \alpha_{i_3}$, 
$\delta_{i_1,i_2,i_3,i_4} = \alpha_{i_1}\cdot \alpha_{i_2}\cdot \alpha_{i_3}\cdot \alpha_{i_4}$.

Observe that each sum over $a \in \{0,1\}^n$ on the RHS is precisely a Sum-Product task over ${\cal C}$, with products ranging from $k=2$ to $k=4$. Therefore we can check that the sum $\sum_{a \in \{0,1\}^n} h(a)$ is zero with $O(s^4)$ calls to Sum-Product over ${\cal C}$. By assumption, this can be done in $O(2^{n-\eps n} \cdot \poly(n,s))$ time.
\end{proof}

The second crucial component yields the ability to solve Circuit Unsatisfiability efficiently with nondeterminism, under the hypotheses (in fact, weaker hypotheses). This is provided by the following lemma, which is similar to (but more complicated than) Lemma 3.1 in \cite{WilliamsJACM14}:

\begin{lemma}\label{nondet-UNSAT} Assume:
\begin{compactitem}
\item There is an $\eps > 0$ and for $k=1,\ldots,4$ there is an $n^{O(1)}\cdot 2^{n-\eps n}$-time algorithm for computing the Sum-Product of $k$ functions from ${\cal C}$.
\item The Circuit Evaluation problem has $\LIN \circ {\cal C}$ circuits of sparsity $n^k$, for some $k > 0$. 
\end{compactitem}
Then there is a nondeterministic $2^{n-\eps n}\cdot\poly(n,s)$-time algorithm for Circuit Unsatisfiability, on arbitrary fan-in-2 circuits with $n$ inputs and $s$ gates. 
\end{lemma}

\begin{proof} Suppose we are given a circuit $C$ with $n$ inputs and $s$ gates of fan-in 2, and wish to nondeterministically prove it is unsatisfiable. Let us index the gates in topological order, so that gates $1,\ldots,n$ are the input gates, and the $s$-th gate is the output gate.

Our nondeterministic algorithm begins by guessing a $\LIN \circ {\cal C}$ circuit $EVAL$ with $n+O(\log s)$ inputs and sparsity at most $(n+s)^{k+1}$, which is intended to encode the Circuit Evaluation function:
\[EVAL(C,x,i) := \text{Evaluate $C$ on $x$, and output the value of the $i$-th gate of $C$.}\] (Note $i$ is encoded as an $O(\log s)$-bit string.) Let \[D(x,i) := EVAL(C,x,i),\] i.e., we think of $C$ as hard-coded in the function, to simplify the notation. Applying Theorem~\ref{check-boolean}, we can check that $D$ encodes a Boolean function in $2^{n-\eps n}\cdot \poly(s,n)$ time.  

Next, we check that $D(a,s) = 0$ for all $a \in \{0,1\}^n$; in other words, $D$ claims that $C$ outputs $0$ on every input. Suppose $D$ has the form
\[D(x,i) = \sum_{j=1}^{(n+s)^{k+1}} \alpha_j \cdot c_j(x,i),\] for some $\alpha_j \in \R$ and $c_j \in {\cal C}$. Since $D$ has already been determined to be Boolean, it suffices to compute $\sum_{a \in \{0,1\}^n} D(a,s)$ to know whether or not $D(x,s)=0$ for all $a$. By exchanging the order of summation, \begin{align*}
\sum_{a \in \{0,1\}^n} D(a,s) &= \sum_{a \in \{0,1\}^n} \left(\sum_j \alpha_j \cdot c_j(a,i)\right) \\
&= \sum_j  \alpha_j \cdot \left(\sum_{a \in \{0,1\}^n} c_j(a,i)\right).
\end{align*}
Therefore we only need to make $(n+s)^{k+1}$ calls to Sum-Product over ${\cal C}$ (with $k=1$) to determine that $D(x,s) = 0$ for all $a \in \{0,1\}^n$. This can be done in $2^{n-\eps n}\cdot \poly(n,s)$ time, by assumption.

Next, we have to check that for every gate $i=1,\ldots,s$, and every $a \in \{0,1\}^n$, $D(a,i)$ correctly reports the output of the $i$-th gate when $C$ evaluates $a$. To check the input gates, we need to check that $D(x,i) = x_i$ for all $i=1,\ldots,n$; we can do this by checking that \[\sum_{a \in \{0,1\}^n} (D(x,i)-x_i)^2 = 0,\] which (by distributivity and re-arranging the order of summation, as in the proof of Theorem~\ref{check-boolean}) can be computed with $O((n+s)^{2(k+1)})$ calls to Sum-Product over ${\cal C}$ (with $k=2$) in $2^{n-\eps n}\cdot \poly(n,s)$ time.

For all gates $i$ other than the input gates, the $i$th-gate takes inputs from previous gates indexed by some $i_1 < i$ and $i_2 < i$, and computes a function of their two outputs. To check the consistency of gate $i$, we can form a degree-3 polynomial $p_i(A,B,C)$ which outputs 0-1 values on all $A,B,C \in \{0,1\}$, such that $p_i(A,B,C) = 0$ if and only if $A$ is the output of gate $i$, given that $B$ is the output of gate $i_1$ and $C$ is the output of gate $i_2$. 

Since $D$ is Boolean-valued, we have reduced our problem to determining that \[\sum_{a \in \{0,1\}^n} p(D(a,i),D(a,i_1),D(a,i_2)) = 0,\] for each gate $i=n+1,\ldots,s$, and each gate $i$'s corresponding input gates $i_1$ and $i_2$. Applying the distributive law to the LHS and exchanging the order of summation (as before), this results in $O((n+s)^{3(k+1)})$ Sum-Product-over-${\cal C}$ computations with up to $k=3$ products, computable in $2^{n-\eps n}\cdot \poly(n,s)$ time.

Our nondeterministic algorithm determines that the input circuit $C$ is unsatisfiable if and only if all of the above checks pass. If $C$ is satisfiable, then every possible $D$ guessed will fail some check. If $C$ is unsatisfiable, then under the hypotheses of the theorem, a $\LIN \circ {\cal C}$ circuit $D$ simulating every gate of $C$ always exists. By guessing this $D$, and running the assumed Sum-Product algorithm, our nondeterministic algorithm accepts.
\end{proof}

After the above preparation, we turn back to the proof of Theorem~\ref{generic-LBs}. At this point, it is simply a matter of applying the above Lemma~\ref{nondet-UNSAT} with the known algorithms-to-lower-bound connections:

\begin{proofof}{Theorem~\ref{generic-LBs}} 
Suppose every $C \in {\cal C}$ has a $\poly(n)$-bit representation, where each $C$ can be evaluated on a given input in $\poly(n)$ time. Recall the hypothesis of the theorem is:

{\narrower 

(A) There is an $\eps > 0$ and for $k=1,\ldots,4$ there is an $n^{O(1)}\cdot 2^{n-\eps n}$-time algorithm for computing the Sum-Product of $k$ functions $f_1(x_1,\ldots,x_n),\ldots,f_k(x_1,\ldots,x_n)$ from ${\cal C}$. 

}

Furthermore, recall that Lemma~\ref{nondet-UNSAT} states:

{\narrower 

Assuming (A) and assuming Circuit Evaluation has $\LIN \circ {\cal C}$ circuits of sparsity $n^k$ for some $k$, there is a nondeterministic $2^{n-\eps n}\cdot\poly(n,s)$-time algorithm for Circuit Unsatisfiability, on arbitrary fan-in-2 circuits with $n$ inputs and $s$ gates. 

}

We can then prove the lower bounds of the theorem readily, as follows.
\begin{compactenum}
\item[(1)] Assume every function in $\NP$ has $\LIN \circ {\cal C}$ circuits of $n^k$ sparsity circuits, for some fixed $k$. Then both hypotheses of Lemma~\ref{nondet-UNSAT} are satisfied (note Circuit Evaluation is in $\P$), and the conclusion implies that there is an $\eps > 0$ such that Circuit Unsatisfiability for (fan-in 2) circuits with $n$ inputs and $2^{\eps n}$ size is solvable in $O(2^{n-\eps n})$ nondeterministic time. Therefore by Theorem~\ref{NP-from-UNSAT}, for every $k$ there is a function in $\NP$ that \emph{does not} have $n^k$-size (fan-in 2) circuits. This is a contradiction because $\LIN \circ {\cal C}$ circuits of $n^k$ sparsity can be simulated with $n^{ck}$-size fan-in-2 circuits, for some universal $c$. 

\item[(2)] The same argument as in (1) and (2) (but with Theorem~\ref{NP-from-UNSAT2} applied) shows that for every unbounded $\alpha(n)$ such that $n^{\alpha(n)}$ is time-constructible, there is a function in $\NTIME[n^{\alpha(n)}]$ that does not have $\LIN \circ {\cal C}$ circuits of polynomial sparsity. 
\end{compactenum}
\end{proofof}

\paragraph{A Note on Lower Bounds for Linear Combinations of ACC Circuits.} Other lower bound consequences of the arguments in Theorem~\ref{generic-LBs} follow easily from combining known results. Here is an example:

\begin{reminder}{Theorem~\ref{lin-acc-thr}} For every $d,m \geq 1$, there is a $b \geq 1$ and an $f \in {\sf NTIME}[n^{\log^b n}]$ that does not have $\LIN \circ \AC^0_d[m] \circ \THR$ circuits of $n^a$ size, for every $a$.
\end{reminder}

This lower bound can be obtained as follows. First, the argument of Lemma~\ref{nondet-UNSAT} also shows:

\begin{theorem} \label{nqp-generic} Assume
\begin{compactitem}
\item There is an $\eps > 0$ and for $k=1,\ldots,4$ there is an $n^{O(1)}\cdot 2^{n-n^{\eps}}$-time algorithm for computing the Sum-Product of $k$ functions from ${\cal C}$.
\item The Circuit Evaluation problem has $\LIN \circ {\cal C}$ circuits of sparsity $n^{a}$, for some $a > 0$. 
\end{compactitem}
Then there is a nondeterministic $2^{n-n^{\eps}}\cdot\poly(n,s)$-time algorithm for Circuit Unsatisfiability, on arbitrary fan-in-2 circuits with $n$ inputs and $s$ gates. 
\end{theorem}

Now we combine this theorem with the following two facts:

\begin{compactenum}
\item For every depth $d$ and integer $m \geq 2$, there is an $\eps > 0$ such that the Sum-Product of $O(1)$ $\AC^0_d[m] \circ \THR$ circuits of $2^{n^{\eps}}$ size can be computed in $2^{n-n^{\eps}}$ time. This simply applies the algorithm for counting satisfying assignments of $\AC^0_d[m] \circ \THR$ circuits (\cite{WilliamsTHR14}).
\item If for some $\alpha > 0$ there is a nondeterministic $2^{n-n^{\alpha}}$-time Circuit Unsatisfiability algorithm for $2^{n^{\alpha}}$-size circuits, then for every $a \geq 1$, there is a $b \geq 1$ such that ${\sf NTIME}[n^{\log^b n}]$ does not have $n^{\log^a n}$-size circuits (this is a theorem of Murray and Williams~\cite{Murray-Williams17}).
\end{compactenum}

Theorem~\ref{lin-acc-thr} is immediate: Assuming ${\sf NTIME}[n^{\log^b n}]$ has $\LIN \circ \AC^0_d[m] \circ \THR$ circuits of $n^a$ size for some $a\geq 1$, both hypotheses of Theorem~\ref{nqp-generic} are satisfied for ${\cal C} = \AC^0_d[m] \circ \THR$, and the conclusion of Theorem~\ref{nqp-generic} combined with item 2 above yields a contradiction.

\subsection{Lower Bounds for Exponential Time With an NP Oracle}

For classes ${\cal C}$ with a natural closure property, the lower bounds can be extended to $2^{\Omega(n)}$ sparsity for a function in $\E^{\NP}$. Recall ${\sf ANY}_c$ denotes the class of Boolean functions with $c$ inputs (the class contains ``any'' such function).

For an integer $c \geq 1$, we say that ${\cal C}$ is \emph{efficiently closed under $\NC^0_c$} if there is a polynomial-time algorithm $A$ such that, given any circuit $C$ of the form ${\cal C} \circ {\sf ANY}_c$, algorithm $A$ outputs an equivalent circuit $D$ from ${\cal C}$ (which is only polynomially larger). We note this property is true of $O(1)$-degree polynomials: 

\begin{proposition}
For every integer $m \geq 2$ and $c \geq 1$, the class ${\cal C} = \bigcup_{d \geq 1}\MODm \circ \AND_d$ is efficiently closed under $\NC^0_c$.
\end{proposition}

\begin{proof} Every $\MODm \circ \AND_d \circ {\sf ANY}_c$ circuit can be represented by an $\MODm \circ \AND_{dc}$ circuit. In particular, every Boolean function on $c$ inputs has an exact representation as a sum (modulo $m$) of ANDs of fan-in $c$; composing such a sum with a $\MODm \circ \AND$ circuit and applying the distributive law yields the result.
\end{proof}

\begin{theorem}\label{generic-LBs2}
There is a universal $c \geq 1$ satisfying the following. Suppose ${\cal C}$ is efficiently closed under $\NC^0_c$, and suppose every $C \in {\cal C}$ has a $\poly(n)$-bit representation, where each $C$ can be evaluated on a given input in $\poly(n)$ time.\\ Assume there is an $\eps > 0$ and for $k=1,\ldots,4$ there is an $n^{O(1)}\cdot 2^{n-\eps n}$-time algorithm for computing the Sum-Product of $k$ functions $f_1(x_1,\ldots,x_n),\ldots,f_k(x_1,\ldots,x_n)$ from ${\cal C}$.\\ Then there is a function in $\E^{\NP}$ that does not have $\LIN \circ {\cal C}$ circuits of sparsity $2^{\alpha n}$, for some $\alpha > 0$.
\end{theorem}

The remainder of this section sketches the proof of Theorem~\ref{generic-LBs2}; we give only a sketch, as the argument closely resembles others~\cite{WilliamsJACM14,JMV13}). 

Let $\eps \in (0,1)$. Assume ${\cal C}$ is efficiently closed under $\NC^0$, and
\begin{compactitem}
\item[(A)] There is an $\eps > 0$ and an $O(2^{n-\eps n})$-time algorithm for computing the Sum-Product of $k$ functions from ${\cal C}$, \emph{and}
\item[(B)] For all functions $f \in \TIME[2^{O(n)}]^{\NP}$ and all $\alpha > 0$, $f$ has $\LIN \circ {\cal C}$ circuits of sparsity $2^{\alpha n}$.
\end{compactitem}

We wish to establish a contradiction. In particular, we will show that assumptions (A) and (B) together imply that every problem in $\NTIME[2^n]$ can be simulated by a nondeterministic $o(2^n)$-time algorithm, contradicting the (strong) nondeterministic time hierarchy theorem~\cite{Seiferas-Fischer-Meyer78,Zak}. 

Let $L \in \NTIME[2^n]$. On a given input $x$, our nondeterministic $o(2^n)$-time algorithm for $L$ has two parts:
\begin{compactitem}
\item[(i)] It \emph{guesses a witness for $x$} of $o(2^n)$ size.
\item[(ii)] It \emph{verifies that witness for $x$} in $o(2^n)$ time.
\end{compactitem}
To handle (i), we use assumption (B) to show that one can nondeterministically \emph{guess} a $2^{\alpha n}\cdot \poly(n)$-size $\LIN \circ {\cal C}$ circuit that encodes a witness for $x$, applying a simple ``easy witness'' lemma from~\cite{Williams10}:

\begin{lemma}[Lemma 3.2 in~\cite{Williams10}] Let ${\cal D}$ be any class of circuits. If $\E^{\NP}$ has circuits of size $S(n)$ from class ${\cal D}$, then for every $L \in \NTIME[2^n]$ and every verifier $V$ for $L$, and every $x \in L$ of length $n=|x|$, there is a $y$ of length $O(2^n)$ such that $V(x,y)$ accepts and the ${\cal D}$-circuit complexity of $y$ (construed as a function $f : \{0,1\}^{n+O(1)} \rightarrow \{0,1\}$) is at most $S(n)$. 
\end{lemma}

In other words, assumption (B) implies that every yes-instance of $L$ has $S(n)$-size ``witness circuits'': a witness of length $O(2^n)$ that can be represented as an $S(n)$-size $\LIN \circ {\cal C}$ Boolean-valued circuit. Furthermore, this holds for every verifier for $L$. 

To handle (ii), we choose an appropriate verifier, so that verifying witnesses becomes equivalent to a simple Sum-Product call. In particular we use the following extremely ``local'' reduction from $L \in \NTIME[2^n]$ to 3SAT instances of $2^n \cdot \poly(n)$ length:

\begin{lemma}[\cite{JMV13}]\label{Succinct-3SAT}
Every $L \in \NTIME[2^n]$ can be reduced to 3SAT instances of $O(2^n \cdot n^4)$ size. Moreover, there is an algorithm that, given an instance $x$ of $L$ and an integer $i \in [O(2^n \cdot n^4)]$ in binary, {\bf reads only $O(1)$ bits of $x$} and outputs the $i$-th clause of the resulting 3SAT formula, in $O(n^4)$ time. 
\end{lemma}

Since in Lemma~\ref{Succinct-3SAT} each bit of the output is a function of some $c \leq O(1)$ inputs, each bit of the output is a member of ${\sf ANY}_c$. So for every instance $x$ of length $n$ for the language $L$, we can produce (in deterministic $\poly(n)$ time) a circuit $D_x$ which is an ordered collection of $O(n)$ functions from ${\sf ANY}_c$. The circuit $D_x$ takes $n+O(\log n)$ binary inputs, construes that input as an integer $i$, and outputs the $i$-th clause of a formula $F_x$ which is satisfiable if and only if $x \in L$. 

Our nondeterministic algorithm for $L$ guesses a $2^{O(\alpha n)}$-sparse $\LIN \circ {\cal C}$ circuit $C_x$ that takes $n+O(\log n)$ inputs and is meant to encode a satisfying assignment for the formula $F_x$. 
We can check $C_x$ is Boolean-valued on all $2^n \cdot \poly(n)$ inputs in $2^{n-\eps n/2}$ time, by applying Theorem~\ref{check-boolean} and letting $\alpha > 0$ be sufficiently small.

Composing $C_x$ with the $O(n)$ polynomials forming $D_x$, we obtain a $2^{O(\alpha n)}$-sparse $\LIN \circ {\cal C} \circ {\sf ANY}_c$ circuit $E$ with $n+O(\log n)$ inputs (composed of three copies of $C_x$, and $O(n)$ copies of $D_x$) such that 
\[\text{$E$ is unsatisfiable if and only if $C_x$ encodes a satisfying assignment for $F_x$.}\] (We leave out the details, as they are provided in multiple other papers~\cite{Williams10,WilliamsJACM14}.) To complete the $o(2^n)$-time algorithm for $L$, it suffices to check unsatisfiability of the resulting $2^{O(\alpha n)}$-size circuit $E$ in $o(2^n)$ nondeterministic time. This would yield the desired contradiction. 

Such a nondeterministic UNSAT algorithm is provided by first converting $E$ into an $\LIN \circ {\cal C}$ circuit in $2^{O(\alpha n)}$ time (using the fact that ${\cal C}$ is efficiently closed under $\NC^0$). This yields a sum of $2^{O(\alpha n)}$ ${\cal C}$-circuits. Analogously to the proof of Lemma~\ref{nondet-UNSAT}, checking the unsatisfiability of such an $E$ can be reduced to $2^{O(\alpha n)}$ calls to Sum-Product of ${\cal C}$, by applying distributivity. Applying the Sum-Product algorithm of assumption (A) that runs in $O(2^{n-\eps n})$ time, and setting $\alpha > 0$ to be sufficiently small, the running time is $o(2^n)$.  

This completes the proof of Theorem~\ref{generic-LBs2}.

\section{Sparse Combinations of Threshold Functions}\label{section-SUM-THR}

We now turn to proving $\LIN \circ \THR$ lower bounds. Due to Lemma~\ref{generic-LBs}, it suffices to give a $2^{n-\eps n}$-time algorithm for the Sum-Product Problem over $\THR$:

{\narrower

{\bf Sum-Product over $\THR$:} Given $k$ linear threshold functions $f_1,\ldots,f_k$, each on Boolean variables $x_1,\ldots,x_n$, compute \[\sum_{x \in \{0,1\}^n} \prod_{i=1}^k f_i(x).\]

}

Putting together various pieces (described in the Preliminaries), there is a substantially faster-than-$2^n$ time algorithm:

\begin{theorem} \label{sum-prod-THR} The Sum-Product of $k$ linear threshold functions on $n$ variables (with weights in $[-n^n,n^n]$) can be computed in $2^{n/2} \cdot n^{O(k)}$ time. 
\end{theorem}

Note that having weights in $[-n^n,n^n]$ is without loss of generality (in our lower bound proofs, our nondeterministic algorithm can always guess an equivalent circuit with such weights, as described by Proposition~\ref{prop-weights}).

\begin{proof} Let $f_1,\ldots,f_k$ be $n$-variable threshold functions. Applying Theorem~\ref{THR2ETHR}, we can write each $f_i$ as a sum of $t=\poly(n)$ \emph{exact} threshold functions: \[f_i(x) = \sum_{i=1}^t g_i(x),\] where each $g_i(x)$ is defined by some weights $w_{i,1},\ldots,w_{i,n} \in \R$ and a threshold value $t \in \R$. 
Therefore we can write the product $f_1\cdots f_k$ as \[\prod_{i=1}^k f_i = \sum_{(i_1,\ldots,i_k)\in[t]^k} g_{i_1}\cdots g_{i_k}.\] 
Each term $g_{i_1}\cdots g_{i_k}$ is a conjunction of $k$ exact thresholds. Applying Theorem~\ref{ANDETHR2ETHR}, each such term can be replaced with a single exact threshold gate, with weights of magnitude $n^{O(kn)}$, i.e., each weight is representable with $O(kn \log n)$ bits. Thus 
\[\prod_{i=1}^k f_i = \sum_{(i_1,\ldots,i_k)\in[t]^k} h_{i_1,\ldots,i_k}\] for some exact threshold gates $h_{i_1,\ldots,i_k}$. The desired sum can therefore be written as
\begin{align*}
\sum_{a \in \{0,1\}^n} \prod_{i=1}^k f_i(a) &= \sum_{a \in \{0,1\}^n} \sum_{(i_1,\ldots,i_k)\in[t]^k} h_{i_1,\ldots,i_k}(a)\\
&= \sum_{(i_1,\ldots,i_k)\in[t]^k} \left(\sum_{a \in \{0,1\}^n} h_{i_1,\ldots,i_k}(a)\right).
\end{align*}
Now observe that each sum $\sum_{a \in \{0,1\}^n} h_{i_1,\ldots,i_k}(a)$ on the RHS is equivalent to an instance of $\#$Subset Sum. In particular, each such sum is counting the number of subsets of a given set of $n$ weights in $[-n^{\Omega(kn)}, n^{O(kn)}]$ which sum to zero. By Theorem~\ref{subset-sum}, this can be computed in $\poly(k,n) \cdot 2^{n/2}$ time.  Since there are $n^{O(k)}$ such sums to compute in the outer sum, the total running time is $n^{O(k)} \cdot 2^{n/2}$. \end{proof} 

The following are immediate from Theorem~\ref{generic-LBs}:

\begin{reminder}{Theorem~\ref{NP-LIN-THR}} 
For all $k$, there is an $f_k \in \NP$ without $\LIN \circ \THR$ circuits of $n^k$ sparsity. Furthermore, for every unbounded $\alpha(n)$ such that $n^{\alpha(n)}$ is time constructible, there is a function in $\NTIME[n^{\alpha(n)}]$ that does not have $\LIN \circ \THR$ circuits of polynomial sparsity.
\end{reminder}

\section{Sparse Combinations of ReLU Gates}\label{section-ReLU}

Recall that a function $f : \{0,1\}^n \rightarrow \R$ from the class ${\sf ReLU}$ is defined with respect to a weight vector $w \in \R^n$ and a scalar $a \in R$, such that for all $a \in \{0,1\}^n$, \[f(x) = \max\{0,\langle w,x\rangle + a\}.\] 
To prove $\LIN \circ {\sf ReLU}$ lower bounds, we give a $2^{n-\eps n}$-time algorithm for the Sum-Product Problem over ${\sf ReLU}$:

{\narrower

{\bf Sum-Product over ${\sf ReLU}$:} Given $k$ ReLU functions $f_1,\ldots,f_k$, each on Boolean variables $x_1,\ldots,x_n$, compute \[\sum_{x \in \{0,1\}^n} \prod_{i=1}^k f_i(x).\]

}

\begin{theorem}\label{sum-prod-ReLU} The Sum-Product of $k$ ${\sf ReLU}$ functions on $n$ variables (with weights in $[-W,W]$) can be computed in $2^{n/2} \cdot n^{O(k)}\cdot \poly(k,n,\log W)$ time. 
\end{theorem}

The proof is similar in spirit to the algorithm for Sum-Product of threshold functions (Theorem~\ref{sum-prod-THR}), except that complications arise due to the real-valued outputs of ${\sf ReLU}$ functions. We end up having to solve a problem generalizing $\#$Subset Sum, but which turns out to have a nice ``split-and-list'' $2^{n/2}$-time  algorithm, analogously to $\#$Subset Sum.

\begin{proof} Let $f_1,\ldots,f_k$ be $n$-variable ReLU functions, defined by weight vectors $w_1,\ldots,w_k \in \R^n$ and scalars $a_1,\ldots,a_k \in \R$, respectively. Our task is to compute
\[\sum_{x \in \{0,1\}^n} \max\{0,\ip{x}{w_1}+a_1\} \cdots \max\{0,\ip{x}{w_k}+a_k\}.\]
First, we note the above sum is equal to 
\[\sum_{x \in \{0,1\}^n} [\ip{x}{w_1} \geq -a_1]\cdot (\ip{x}{w_1}+a_1) \cdots [\ip{x}{w_k} \geq -a_k]\cdot (\ip{x}{w_k}+a_k),\] where we are using the Iverson bracket notation $[P]$ to denote a function that outputs $1$ if $P$ is true and $0$ otherwise.
Applying Theorem~\ref{THR2ETHR}, each of the threshold functions $[\ip{x}{w_i} \geq -a_i]$ can be represented as a linear combination of $t=\poly(n)$ exact threshold functions. In particular there are exact thresholds $g_{i,j}$ such that the above sum equals
\[\sum_x \left(\sum_{j=1}^t g_{1,j}(x)\right)\cdot (\ip{x}{w_1}+a_1) \cdots\left(\sum_{j=1}^t g_{k,j}(x)\right)\cdot (\ip{x}{w_k}+a_k).\]
Applying the distributive law, the above sum equals 
\[\sum_x \sum_{j_1,\ldots,j_k \in [t]^k} g_{1,j_1}(x) \cdots g_{k,j_k}(x)\cdot (\ip{x}{w_1}+a_1)\cdots (\ip{x}{w_k}+a_k).\] Re-arranging the summation order yields
\[\sum_{j_1,\ldots,j_k \in [t]^k} \left(\sum_x g_{1,j_1}(x) \cdots g_{k,j_k}(x)\cdot (\ip{x}{w_1}+a_1)\cdots (\ip{x}{w_k}+a_k)\right).\]
Applying Theorem~\ref{ANDETHR2ETHR}, each $g_{1,j_1}(x) \cdots g_{k,j_k}(x)$ can be replaced by a single exact threshold $h_{j_1,\ldots,j_k}(x)$. 

Our task has been reduced to $n^{O(k)}$ computations of the form
\begin{align}\label{ethr-ips}
\sum_{x \in \{0,1\}^n} h_{j_1,\ldots,j_k}(x) \cdot (\ip{x}{w_1}+a_1)\cdots (\ip{x}{w_k}+a_k).
\end{align}
Without the $(\ip{x}{w_1}+a_1)\cdots (\ip{x}{w_k}+a_k)$ term, \eqref{ethr-ips} would be exactly a $\#$Subset Sum instance, as in Theorem~\ref{sum-prod-THR}. In this new situation, we need to count a ``weighted'' sum over the subset sum solutions, where the weights are determined by a product of $k$ inner products of the solution vectors with some fixed vectors.

Let us now describe how to solve the generalized problem given by \eqref{ethr-ips}. To keep the exposition clear, we will walk through an attempted solution and fix it as it breaks. 

Suppose the exact threshold function $h_{j_1,\ldots,j_k}(x)$ of \eqref{ethr-ips} is defined by weights $\alpha_1,\ldots,\alpha_n \in \R$ and threshold value $t \in \R$, so that \[h_{j_1,\ldots,j_k}(x) = 1 \iff \sum_{i=1}^n \alpha_i x_i = t.\] As with the Subset Sum problem, we begin by splitting the set of variables $x$ into two halves, $\{x_1,\ldots,x_{n/2}\}$ and $\{x_{n/2+1},\ldots,x_n\}$ (WLOG, assume $n$ is even). Correspondingly, we split each of the $k$ weight vectors $w_i \in \R^n$ of \eqref{ethr-ips} into two halves, $w^{(1)}_i \in \R^{n/2}$ and $w^{(2)}_i \in \R^{n/2}$ for the first and second halves of variables, respectively.

We list all $2^{n/2}$ partial assignments to the first half, and all $2^{n/2}$ partial assignments to the second. For each partial assignment $A = (A_1,\ldots,A_{n/2})$ to the first half of variables $\{x_1,\ldots,x_{n/2}\}$, we compute a vector $v_A$, as follows:
\begin{compactitem}
\item $v_A[0] := -t + \sum_{i=1}^{n/2} \alpha_i A_i$,
\item for all $j=1,\ldots,k$, $v_A[j] := a_j + \ip{w^{(1)}_j}{(A_1,\ldots,A_{n/2})}$.
\end{compactitem}
For each partial assignment $A'=(A_{n/2+1},\ldots,A_n)$ from the second half, we compute a vector $w_{A'}$:
\begin{compactitem}
\item $w_{A'}[0] :=  \sum_{i=n/2+1}^{n} \alpha_i A_i$,
\item for all $j=1,\ldots,k$, $w_{A'}[j] := \ip{w^{(2)}_j}{(A_{n/2+1},\ldots,A_n)}$.
\end{compactitem}
Notice that $v_A[0]+w_{A'}[0] = 0$ if and only if $h_{j_1,\ldots,j_k}(A,A') = 1$. Thus in our sum, we only need to consider pairs of vectors $v_A$ from the first half and vectors $w_{A'}$ from the second half such that $v_A[0]+w_{A'}[0]=0$. Moreover, note that for all $j=1,\ldots,k$, 
\[v_A[j]+w_{A'}[j] = \ip{x}{w_j}+a_j.\] 
It follows that \eqref{ethr-ips} equals
\[\sum_{(v_A,w_{A'}) ~:~ v_A[0]+w_{A'}[0]=0} (v_A[1]+w_{A'}[1])\cdots (v_A[k]+w_{A'}[k]).\] 

The Subset-Sum algorithm of Horowitz and Sahni~\cite{Horowitz-Sahni74} shows how to efficiently find pairs $(v_A,w_{A'})$ with $v_A[0]+w_{A'}[0]=0$: sorting all vectors in the second half by their $0$-th coordinate, for each vector $v_A$ from the first half we can compute (in $\poly(n)$ time) the number of second-half vectors $w_{A'}$ satisfying $v_A[0]+w_{A'}[0]=0$ (even if there are exponentially many such vectors). However it is unclear how to incorporate the odd-looking $(v_A[1]+w_{A'}[1])\cdots (v_A[k]+w_{A'}[k])$ multiplicative factors into a weighted sum.

To do so, we modify the vectors $v_A$ and $w_B$ as follows. Consider the expansion of $\prod_{i=1}^k (v_A[i]+w_{A'}[i])$ into a sum of $2^k$ products: it can be seen as the inner product of two $2^k$-dimensional vectors, where one vector's entries is a function solely of $v_A$ and the other vector's entries is a function solely of $w_{A'}$. (Furthermore, note that the number of bits needed to describe entries in these new vectors has increased only by a multiplicative factor of $k$.) 

Thus we can assign $(2^k+1)$-dimensional vectors $v'_A$ (in place of the $v_A$) and $w'_B$ (in place of the $w_B$) such that $v'_A[0]=v_A[0]$, $w'_A[0]=w_A[0]$, and for all $A,A'$ we have
\[(v_A[1]+w_{A'}[1])\cdots (v_A[k]+w_{A'}[k]) = \sum_{j=1}^{2^k} v'_A[j]\cdot w'_{A'}[j].\] 
Now our goal is to compute
\begin{align}\label{goal-sum}
\sum_{(v'_A,w'_{A'}) ~:~ v'_A[0]+w'_{A'}[0]=0}
\left(\sum_{j=1}^{2^k} v'_A[j]\cdot w'_{A'}[j]\right).
\end{align}
We can get a more efficient algorithm for the problem defined by \eqref{goal-sum}, by preprocessing the second half of vectors (i.e., the $w'_{A'}$ vectors). For each distinct value $e = w'_A[0] \in \R$ among the $2^{n/2}$ vectors in the second half, we make a new $(2^k+1)$-dimensional vector $W'_e$ where:
\begin{compactitem}
\item $W'_e[0] = e$, and 
\item for all $i=1,\ldots,2^k$, $W'_e[i] = \sum_{w'_A ~:~ w'_A[0]=e} w'_A[i]$.
\end{compactitem}

That is, the coordinates $1,\ldots,2^k$ of $W'_e$ are obtained by component-wise summing all vectors $w'_A$ such that $w'_A[0]=e$. The preparation of the vectors $W'_e$ can be done in $2^{n/2}\cdot \poly(k,n,\log W)$ time, by partitioning all $2^{n/2}$ vectors $w'_A$ from the second half of variables into equivalence classes (where two vectors are equivalent if their $0$-coordinates are equal), then obtaining each $W'_e$ by summing the vectors in one equivalence class. 

Finally, we can use the $W'_{A'}$ vectors to compute the sum~\eqref{goal-sum} in $2^{n/2}\cdot 2^k \cdot \poly(k,n,\log W)$ time. Have a running sum that is initially $0$. Iterate through each vector $v'_A$ from the first half of variables, look up the corresponding second-half vector $W'_e$ (with $v'_A[0] = -W'_e[0]$) in $\poly(k,n,\log W)$ time, and add the inner product \[\sum_{i=1}^{2^k} v'_A[i]\cdot W'_e[i]\] to the running sum.
Because each vector $(W'_e[1],\ldots,W'_e[2^k])$ is the sum of \emph{all} vectors $(w'_{A'}[1],\ldots,w'_{A'}[2^k])$ such that $v'_A[0]+w'_{A'}[0]=0$, each inner product $\sum_{i=1}^{2^k} v'_A[i]\cdot W'_e[i]$ contributes 
\[\sum_{w'_{A'} ~:~ v'_A[0]+w'_{A'}[0]=0} \left(\sum_{j=1}^{2^k} v'_A[j]\cdot w'_A[j]\right)\] to the running sum. Therefore after iterating through all vectors $v'_A$, our running sum has computed \eqref{goal-sum} exactly, in only $2^{n/2}\cdot 2^k\cdot \poly(n,\log W)$ time.
\end{proof}

From the algorithm of Theorem~\ref{sum-prod-ReLU}, we immediately obtain the $\LIN\circ {\sf ReLU}$ lower bounds of Theorem~\ref{NP-ReLU}.

\section{Sparse Combinations of Low-Degree Polynomials over Finite Fields}\label{section-low-degree}

We can also prove lower bounds for linear combinations of low-degree $\F_p$-polynomials in $n$ variables, for any prime $p$, by giving a faster Sum-Product algorithm. In this context, the Sum-Product problem becomes:

{\narrower 

{\bf Sum-Product over $\MOD_p\circ\AND_d$:} Given $k$ polynomials $p_1,\ldots,p_k \in \F_p[x_1,\ldots,x_n]$, each of degree at most $d$, compute \[\sum_{x \in \{0,1\}^n} \left(\prod_{i=1}^k p_i(x)\right),\] where the sum over all $x \in \{0,1\}^n$ is taken over the reals (or rationals). 

}

That is, we treat each $\prod_{i=1}^k p_i(x)$ as a function from $\{0,1\}^n$ to $\{0,1,\ldots,p-1\} \subset \Q$, and wish to compute the sum of these integers over all $x \in \{0,1\}^n$.

In related work, Lokshtanov \emph{et al.}~\cite{LokshtanovPTWY17} showed how to (deterministically) count solutions in $\F_p^n$ to a system of $\ell$ degree-$d$ $\F_p$-polynomials in $p^{n+o(n)-n/O(dp^{6/7})}\cdot \poly(\ell)$ time. For our Sum-Product problem, we need to compute a ``weighted'' sum (the terms can take on values in $\{0,\ldots,p-1\}$), and we need to count the weighted sum over only \emph{Boolean} assignments. We can achieve this, with a comparable runtime savings involving $k$ and $p$:

\begin{theorem}\label{sum-prod-polys}
The Sum-Product of $k$ degree-$d$ polynomials $p_1,\ldots,p_k \in \F_p[x_1,\ldots,x_n]$ can be computed in $p^{2k} \cdot(1.9^n + 2^{n-n/(6dp)})\cdot \poly(n)$ time. 
\end{theorem}

\begin{proof} Let $p_1,\ldots,p_k$ be given. We wish to compute 
\begin{align}\label{sp-e}
\sum_{x \in \{0,1\}^n} \left(\prod_{i=1}^k p_i(x)\right),
\end{align} where each product outputs an integer in $\{0,1,\ldots,p-1\}$. We first convert the Sum-Product problem of \eqref{sp-e} to an equivalent sum where each ``term'' in the sum is a small system of polynomial equations. 

We say that a function $f : \{0,1\}^n \rightarrow \{0,1\}$ is an \emph{exact $\F_p$-polynomial function} if there is a polynomial $p \in \F_p[x_1,\ldots,x_n]$ and $a \in \F_p$ such that for all $x \in \{0,1\}^n$, 
\[f(x) = 1 \iff p(x) = a.\] 
We use the notation $[p(x)=a]$ to denote such an exact polynomial function. Let us replace each polynomial $p_i(x)$ in the sum-product expression with an equivalent linear combination (over $\Z$) of exact polynomial functions. In particular, replace each $p_i(x)$ with the sum \emph{over the integers} \[\sum_{a \in \F_p} a \cdot [p_i(x)=a].\] That is, we are replacing $p_i(a)$ with an equivalent integer-valued sum of $p$ Boolean functions. Now the desired sum \eqref{sp-e} looks like:
\begin{align}\nonumber
\sum_{x \in \{0,1\}^n} \left(\prod_{i=1}^k \left(\sum_{a \in \F_p} a \cdot [p_i(x)=a]\right)\right) &=
\sum_{x \in \{0,1\}^n} \sum_{(a_1,\ldots,a_k) \in \F_p^k} a_1\cdots a_k \cdot \prod_{i=1}^k [p_i(x)=a_i]~~\text{(by distributivity)}\\
&=  \sum_{(a_1,\ldots,a_k) \in \F_p^k} a_1\cdots a_k \cdot \left(\sum_{x \in \{0,1\}^n} [p_1(x)=a_1]\cdots [p_k(x)=a_k]\right).\label{sp-e2}
\end{align}
Each inner sum in \eqref{sp-e2} counts the number of Boolean solutions to a system of polynomial equations $p_1(x)=a_1,\ldots,p_k(x)=a_k$. We can further reduce this problem to counting the number of Boolean solutions to \emph{one} equation, by applying a simple reduction (from~\cite{WilliamsSOSA18}). Namely, we have the equation
\begin{align}\label{sp-e3}
\sum_{x \in \{0,1\}^n} \prod_{i=1}^k [p_i(x)=a_i] 
= \frac{1}{p^k} \sum_{(b_1,\ldots,b_k) \in \F_p^k} \sum_{x\in \{0,1\}^n} \left(\left[\sum_{j=1}^k b_j \cdot (p_j(x)-a_j) = 0\right] - \left[\sum_{j=1}^k b_j \cdot (p_j(x)-a_j) = 1\right]\right).
\end{align}
To see why \eqref{sp-e3} holds, let $x \in \{0,1\}^n$ such that $[p_1(x)=a_1]\cdots [p_k(x)=a_k] = 1$. Then for \emph{every} $(b_1,\ldots,b_k)\in\F_p^k$, we have $[\sum_{j=1}^k b_j \cdot (p_j(x)-a_j) = 0] = 1$. So every solution $x$ to the system of $k$ equations is counted for $p^k$ times in \eqref{sp-e3}; since the result is divided by $p^k$, each solution contributes $1$ to \eqref{sp-e3}. On the other hand, if $x$ is not a solution to the system, and $[p_1(x)=a_1]\cdots [p_k(x)=a_k] = 0$, then for some $j$, $p_j(a) - a_j \neq 0$. It follows that there are precisely $p^{k-1}$ vectors $(b_1,\ldots,b_k) \in \F_p^k$ such that $[\sum_{j=1}^k b_j \cdot (p_j(x)-a_j) = 0]=1$, and there are precisely $p^{k-1}$ (other) vectors $(b'_1,\ldots,b'_k) \in \F_p^k$ such that $[\sum_{j=1}^k b'_j \cdot (p_j(x)-a_j) = 1]=1$. These two equal counts cancel out in the sum of \eqref{sp-e3}, so non-solutions to the system contribute $0$ to the sum of \eqref{sp-e3}.

Putting \eqref{sp-e2} and \eqref{sp-e3} together, the original Sum-Product problem \eqref{sp-e} can now be reduced to the computation of $O(p^{2k})$ sums, each of the form 
\[\sum_{x\in \{0,1\}^n} [q(x_1,\ldots,x_n) = 0],\] where $q$ is an $\F_p$-polynomial of degree at most $d$. That is, to obtain \eqref{sp-e}, we only need to count the Boolean roots of $O(p^{2k})$ polynomials $q$, and take the appropriate $\R$-linear combination of these counts. 

Let us now focus on counting roots to a single polynomial $q(x_1,\ldots,x_n)$ of degree $d$. Let $P_{\ell}(z)$ be the modulus-amplifying polynomial of degree $2\ell-1$, from Theorem~\ref{mod-amplifying}. 
Let $\delta \in (0,1/2)$ be a parameter, and consider the following ``reduced'' polynomial in $n-\delta n$ variables, over the integers:
\[Q(x_1,\ldots,x_{n-\delta n}) := \sum_{a_1,\ldots,a_{\delta n} \in \{0,1\}} P_{\delta n}(1-q(x_1,\ldots,x_{n-\delta n},a_1,\ldots,a_{\delta n})^{p-1}).\] Note that $Q$ has degree less than $2dp\delta n$. Set $\delta = 1/(6dp)$, and note that $2dp\delta n < (n-\delta n)/2$.
Over $\F_p$, the polynomial $1-q(x)^{p-1}$ equals $1 \bmod p$ if $x$ is a root of $q$, and is $0 \bmod p$ otherwise. 
Applying the modulus-amplifying properties of $P_{\delta n}$, we have:
\begin{compactitem}
\item If $x$ is a root of $q$, then $P_{\delta n}(1-q(x)^{p-1}) = 1 \bmod p^{\delta n}$.
\item If $x$ is not a root of $q$, then $P_{\delta n}(1-q(x)^{p-1}) = 0 \bmod p^{\delta n}$.
\end{compactitem}
As the sum in $Q$ is over only $2^{\delta n}$ such $P_{\delta}(\cdots)$ terms, and $p \geq 2$, we conclude that for all $b_1,\ldots,b_{n-\delta n} \in \{0,1\}$, the quantity $(Q(b_1,\ldots,b_{n-\delta n}) \bmod p^{\delta n})$ \emph{equals} the number of $a_1,\ldots,a_{\delta n} \in \{0,1\}$ such that \[q(b_1,\ldots,b_{n-\delta n},a_1,\ldots,a_{\delta n}) = 0.\] Therefore if we evaluate the polynomial $Q$ over all $2^{n-\delta n}$ Boolean assignments $(b_1,\ldots,b_{n-\delta n})$, compute each value separately modulo $p^{\delta n}$, then sum those values over the integers, we will obtain the number of Boolean roots of $q$. 

Over Boolean assignments, we may assume without loss of generality that $Q$ is multilinear (i.e. $x_i^2 = x_i$ for all $i$). Since $2dp\delta n < (n-\delta n)/2$, standard properties of binomial coefficients imply that the number of monomials of $Q$ is \[O\left(\binom{n-\delta n}{2dp\delta n}\right).\]
By constructing $Q$ term-by-term (expanding each $P_{\delta n}(1-q(x_1,\ldots,x_{n-\delta n},a_1,\ldots,a_{\delta n})^{p-1})$ one-by-one, and adding them to a running sum, similar to~\cite{ChanW16,LokshtanovPTWY17}), we may represent $Q$ as a sum of $O\left(\binom{n-\delta n}{2dp\delta n}\right)$ monomials, constructed in $\poly(n) \cdot \binom{n-\delta n}{2dp\delta n}$ time. Letting $\delta = 1/(6dp)$, the number of monomials of $Q$ is less than $\binom{n}{n/3}\leq 1.9^n$. Applying the fast polynomial evaluation algorithm of Theorem~\ref{poly-eval}, $Q$ can be evaluated on all $2^{n-n/(6dp)}$ Boolean assignments in time $(1.9^n + 2^{n-n/(6dp)})\cdot \poly(n)$ time.
\end{proof}

Therefore, for every \emph{fixed} degree $d$ and prime $p$, there is an $\eps > 0$ such that the relevant Sum-Product problem is in $2^{n-\eps n}\cdot \poly(n)$ time. This immediately implies the lower bounds of Theorems~\ref{NP-polys} and \ref{ENP-polys}. In particular, to prove \ref{ENP-polys} we apply Theorem~\ref{generic-LBs2}. Fix an integer degree $d$, and let $c \geq 1$ be the universal constant (from Theorem~\ref{generic-LBs2}) such that we need to solve Sum-Product for $\MODp \circ \AND_d \circ {\sf ANY}_c$ circuits. Converting to $\LIN \circ \MODp \circ \AND_{dc}$, Theorem~\ref{sum-prod-polys} says that the Sum-Product problem can be solved in $2^{n-n/O(dc)}$ time (omitting low-order terms).

\section{Conclusion}

Applying old and new tools, we have established several strong new lower bounds for representing Boolean functions in different regimes. Among the most interesting open problems remaining, we find the Quadratic Uncertainty Principle (that $\AND$ requires a large $\R$-linear combination of quadratic $\F_2$-polynomials) to be especially intriguing. Quadratic polynomials have special properties that higher degrees do not; for example, one can count the roots of a given quadratic $\F_p$-polynomial in \emph{polynomial time} (see~\cite{WilliamsSOSA18} for a recent application of this phenomenon). Therefore in some cases, our $2^{n-\eps n}$-time algorithms become $\poly(n)$-time algorithms. This \emph{should} imply lower bounds for functions \emph{in $\P$} against linear combinations of quadratic $\F_2$-polynomials, perhaps even lower bounds against the AND function, but so far we have not yet been able to prove such bounds. 

A longstanding problem in circuit complexity---seemingly related to the Quadratic Uncertainty Principle---is the Constant Degree Hypothesis of Barrington, Straubing, and Therien~\cite{BST90}:

\begin{hypothesis}[Constant Degree Hypothesis (CDH)] For every constant $d \geq 1$ and primes $p, q$, there is an $\eps > 0$ such that the $\AND$ function on $n$ variables cannot be computed by $\MOD_p \circ \MOD_q \circ \AND_d$ circuits of $2^{\eps n}$ size.
\end{hypothesis}

The CDH is currently only known to be true for $d=1$, and for $p=q$. Can the techniques of this paper say anything about such problems, even for the case of $d=2$? 

\paragraph*{Acknowledgements.} I thank Lijie Chen, Pooya Hatami, Adam Klivans, Shachar Lovett, and Anirbit Mukherjee for comments and discussions on the topics of this paper. In particular, I am grateful to Shachar for noticing a gap in a lemma in an earlier version of this paper. I am also grateful to Brynmor Chapman for his proofreading, and patience with my explanations regarding this paper.

\bibliographystyle{alpha}
\bibliography{cc-papers.bib}

\appendix 

\section{Linear Lower Bound for AND With Sums of Quadratic Polynomials}\label{AND-linear-LB}

For reference, we report a folklore $\Omega(n)$ lower bound on representing AND with linear combinations of quadratic $\F_2$-polynomials (recall it is conjectured that the sparsity lower bound is $2^{\Omega(n)}$). The below proof was communicated to us by Shachar Lovett.

\begin{theorem}[Lovett~\cite{Lovett-personal17}] The $\AND$ function on $n$ inputs does not have $\LIN \circ \MOD2 \circ \AND_2$ circuits of sparsity less than $n/2$.
\end{theorem}

\begin{proof} Let $f : \{0,1\}^n \rightarrow \{0,1\}$ be the NOR function (which by DeMorgan's laws has the same sparsity as AND). Suppose we can write 
\[f(x) = \sum_{i=1}^s \alpha_i (-1)^{q_i(x)},\] where the $q_i(x)$ are quadratic $\F_2$-polynomials, and all $\alpha_i \in \R$. Note that without loss of generality we may assume $q_i(0)=0$ for all $i$ (if $q_i(0) = 1$, then replacing $\alpha_i$ by $-\alpha_i$ and $q_i(x)$ by $q_i(x)+1$ yields an equivalent expression). If $s<n/2$, then by the Chevalley–Warning theorem, the number of common roots of $\{q_1,...,q_r\}$ is divisible by 2. But then there is another common root $x^{\star}$, so $f(0)=f(x^{\star})$, contradicting the definition of NOR.
\end{proof}

\end{document}